\documentclass{amsart}
\usepackage{amsmath,latexsym, graphicx}
\usepackage[psamsfonts]{amssymb}
\usepackage{times}
\usepackage[mathcal]{euscript}
\usepackage{color}
\usepackage{amsfonts}
\usepackage{amssymb}
\usepackage{amsmath}
\usepackage{amsthm}
\usepackage{bm}
\usepackage[english]{babel}
\usepackage[bookmarks]{hyperref}
\numberwithin{equation}{section} \textwidth=140mm \textheight=200mm
\renewcommand{\epsilon}{\varepsilon}

\renewcommand{\epsilon}{\varepsilon}

\renewcommand{\hat}{\widehat }

\newcommand{\black}{\color{black}}

\newcommand{\be}{\begin{equation}}
\newcommand{\ee}{\end{equation}}

\newcommand{\R}{\mathbb{R}}

\newcommand{\T}{\mathbb{T}}

\renewcommand{\det}{\mathop{\mathrm{det}}}

{\bf}{\it}
\newtheorem{theorem}{Theorem}[section]
\newtheorem{lemma}[theorem]{Lemma}

\date{\today}
\begin{document}


\title[The number and location of two particle]{The number and location of two particle Schr\"odinger operators on a lattice}

\author{Sobir Ulashov, Shakhobiddin Khamidov, Shukhrat Lakaev}

\address[S. Ulashov]{Samarkand State University, University boulevard, 15, 140104 Samarkand, Uzbekistan}
\email{sobirulashov19@gmail.com}

\address[Sh. Khamidov]{V.I. Romanovskiy Institute of Mathematics, Uzbekistan Academy of Sciences, Tashkent, 100174, Uzbekistan}
\email{shoh.hamidov1990@mail.ru}

\address[Sh. Lakaev]{National University of Uzbekistan, University street, 4, 100174 Tashkent, Uzbekistan}
\email{shlakaev@mail.ru}

\begin{abstract}
We study the Schr\"odinger operators ${H}_{\lambda\mu}(K)$ with $K\in\T^2$ being the fixed quasimomentum of a pair of particles, associated with a system of two arbitrary particles on a two-dimensional lattice $\mathbb{Z}^2$ with on-site and nearest-neighbor interactions of strengths $\lambda\in\R$
and $\mu\in\R$, respectively.
We divide the $(\lambda,\mu)$-plane of parameters $\lambda$ and $\mu$ into connected components, such that in each component, the Schr\"{o}dinger operator $H_{\lambda\mu}(0)$ has a fixed number of eigenvalues. These eigenvalues are located both below the bottom
of the essential spectrum and above its top. Additionally, we establish a sharp lower bound for the number of isolated eigenvalues of $H_{\lambda\mu}(K)$ within each connected component.

\end{abstract}

\maketitle

\black
\section{\textbf{Introduction}}\label{sec:intro}

Lattice models play a significant role in several branches of physics.
Among such models are the few-body lattice Hamiltonians \cite{ham1},
which can be considered a minimalist version of the corresponding Bose- or
Fermi-Hubbard models involving a fixed finite number of particles of a specific type.
These discrete Hamiltoniand have great theoretical interest on their own \cite{ham2,ham3,ham4,ham5,%
ham6,ham7,ham8,ham9,ham10,ham11,ham111,ham12}.
In addition, these discrete Hamiltonians can be viewed as a natural approximation for their continuous counterparts, allowing us to study few-body  phenomena in the context of the theory of bounded operators, as discussed in \cite{ham13}. A still intriguing phenomenon is the celebrated Efimov effect \cite{ham14}, which has been proven to occur not only in the continuous case, but also in lattice three-body problems. This effect has been studied in several papers, including \cite{ham5,
ham6,ham7,ham141,ham15,ham151}.
Furthermore, the discrete
Schr\"odinger represent a simple and natural model for describing few-body systems composed of particles moving through periodic structures.
This is true, for example, the case for ultracold atoms that are injected into optical crystals created by the interference of two counter-propagating laser beams. This is discussed in the literature\cite{ham16,ham17}.

The study of ultracold few-atom systems in optical lattices has received significant attention in recent years, due to the highly controllable parameters that they possess. These parameters include the lattice geometry, particle mass, two-body interactions, and temperature, which can be precisely adjusted to allow for a wide range of experimental possibilities
(see, for example, \cite{ham16,ham18,ham19,ham20} for more information).

Unlike traditional condensed matter systems, where stable composite objects are typically formed by attractive forces, the controllability of ultracold atomic systems in an optical lattice provides an opportunity to experimentally observe stable repulsive bound pairs of ultracold atoms. For example, see references \cite{ham17,ham21} for more information on this phenomenon. One-particle, one-dimensional lattice Hamiltonians have already been of interest in a variety of applications. For example, in \cite{ham22}, a one-dimensional lattice Hamiltonian for a single particle has been used to demonstrate how the arrangement of molecules of a particular class in certain lattice structures can increase the probability of nuclear fusion.

Unlike in the continuous case, the few-body lattice system does not permit the separation of the center-of-mass motion. Nevertheless, the discrete translational invariance allows us to utilize the Floquet-Bloch decomposition (see, for instance, \cite{ham23}, Section 4).
In particular, the total $n$-particle Hamiltonian $\mathrm{H}$ can be expressed as a direct integral over the (quasi)momentum space:
\begin{equation}\label{U1}
\mathrm{H}\simeq\int\limits_{K\in \mathbb{T}^d} ^\oplus  H(K)\,d
K,
\end{equation}
where $\mathbb{T}^d$ is the $d$-dimensional torus and $H(K)$ is the fiber
Hamiltonian, acting  on the corresponding functional Hilbert space on
$\mathbb{T}^{(n-1)d}$.

Recall that the fiber Hamiltonians $H(K)$ have been studied in various works, where they non-trivially depend on the quasimomentum $K \in \mathbb{T}^d$ (see, for example, \cite{ham1,ham23,ham24,ham25}).
It is well known that the Efimov effect \cite{ham14}, which we have already mentioned before, is originally observed in three-body systems moving in three-dimensional space $\mathbb{R}^3$.

The main idea of the Efimov effect is the following. A system of three particles in ${\mathbb R}^3$ with pairwise attractive short-range potentials has an infinite number of bound states with energies converging exponentially to zero, if the two-body subsystems do not have a negative energy spectrum and at least two of them are resonantly interacting, meaning that any arbitrarily small perturbation of the two-body interaction can produce a negative energy state.
A rigorous mathematical proof of the Efimov effect can be found in the following references \cite{ham26,ham27,ham28}.
In
\cite{ham5,ham6,ham7}, the existence of the Efimov  effect has also been proven for three-body systems on the  three-dimensional lattice $\mathbb{Z}^3$. In the latter case, a mathematical proof is available \cite{ham29}, and the phenomenon has been named the super Efimov effect due to the double exponential convergence of the binding energies towards the three-body threshold.

In this paper, we investigate the way in which new eigenvalues arise from the lower and/or upper thresholds of the essential (continuous) spectrum of the fiber Hamiltonians $H(K)$.
In order to obtain more detailed information, we consider the interaction between particles, which contains two terms. The first term is only non-trivial when the particles are located on the same site of the lattice. The second term is only non-trivial when the particles are positioned on the nearest neighboring sites.

Thus, as the entries $H(K)$ in \eqref{U1}, in this work we study the
family of the fiber Hamiltonians
\begin{equation*}
H_{\lambda\mu}(K):=H_0(K) + V_{\lambda\mu},\qquad
K=(K_1,K_2)\in\mathbb{T}^2,
\end{equation*}
where $H_0(K)$ is the fiber kinetic-energy operator,
$$
\bigl(H_0(K)f\bigr)(p)=\mathcal{E}_K(p)f(p), \quad p=(p_1,p_2)\in\mathbb{T}^2, \quad
 f\in L^{2}(\mathbb{T}^2),
$$
with
\begin{equation*}
\mathcal{E}_K(p):= \varepsilon(p)+\gamma \varepsilon(K-p), \,\,\gamma>0
\end{equation*}
and $V_{\lambda\mu}$ is the interaction potential. The parameters
$\lambda$ and $\mu$, called coupling constants, describe interactions
between the particles which are located on one site and  on
nearest-neighboring-sites of the lattice, respectively.

Within this new model, we find the exact number and locations of the eigenvalues of the operator $H_{\lambda\mu}(0)$, as well as the mechanisms of emission and absorption of these eigenvalues at the thresholds of its essential spectrum. We also establish sharp lower bounds on the number of isolated eigenvalues $H_{\lambda\mu}(K)$, which lie both below the essential spectrum and above it, depending on the quasimomentum $K \in \mathbb{T}^2$. For this, we use the results obtained for $H_{\lambda\mu}(0)$ and the nontrivial dependence of the dispersion relation on the quasimomentum $K \in \mathbb{T}^2$.

\section{ The two-particle Hamiltonian on lattices} \label{sec:hamiltonian}

\subsection{The two-particle Hamiltonian: the position-space representation}\label{subsection_position}

Let ${\mathbb{Z}}^2$ be the two-dimensional integer lattice and $\ell^2([{\mathbb{Z}}^2]^2)$
be the Hilbert space of square-summable functions on the Cartesian product of ${\mathbb{Z}}^2$.

The free Hamiltonian $\hat{H}_0$ of a system consisting of two arbitrary particles (a boson and a fermion) on a two-dimensional lattice $\mathbb{Z}^2$ acts on $\ell^2[(\mathbb{Z}^{2})^2]$ and has the form

 \begin{equation*}
\widehat{\mathbb{H}}_{\,0}=\frac{1}{2m_1}{\Delta}\otimes I +\frac{1}{2m_2}I\otimes \Delta,
\end{equation*}
where $m_\alpha>0,\,\, \alpha=1,2,$ are the masses of the particles.

The lattice Laplacian, denoted by $\Delta$, is a difference operator that describes the movement  of a particle from one site to its nearest-neighbor sites. It is defined as:
$$
(\Delta\hat{\psi})(x)=\sum_{s\in\mathbb{Z}^2,\mid s\mid=1}[\hat{\psi}(x)-\hat{\psi}(x+s)],\,\,\hat{\psi}\in\ell^{2}(\mathbb{Z}^{2}).
$$

We note that $\widehat{\mathbb{H}}_{\,0}$ acts in $\ell^{2}[(\mathbb{Z}^{2})^{2}]$ as
\begin{equation*}
(\widehat{\mathbb{H}}_{\,0} \hat{\psi})(x,y)= \sum_{s_1\in\mathbb{Z}^2} \hat \epsilon(s_1) \hat
{\psi}(x+s_1,y) + \gamma \sum_{s_2\in\mathbb{Z}^2} \hat \epsilon(s_2) \hat {\psi}(x,y+s_2),\,\,\hat {\psi}\in \ell^{2}[(\mathbb{Z}^{2})^{2}],
\end{equation*}
where $\gamma=\frac{m_1}{m_2}$ and the series  $\sum_{s\in\mathbb{Z}^2} \hat \epsilon(s)$ is assummed to be absolutely convergent, that is
\begin{equation*}
\{\hat \epsilon(s)\}_{s\in \mathbb{Z}^2}\in \ell^{1}(\mathbb{Z}^{2}).
\end{equation*}
We also assume that the "self-adjointness" condition is satisfied
$$
\hat \epsilon(s)=\overline{\hat \epsilon(-s)},\,\,s\in \mathbb{Z}^2.
$$

In the position-space representation, the total Hamiltonian for two interacting particles, $\widehat {\mathbb{H}}_{\lambda\mu}$, is a bounded, self-adjoint operator that acts on the space $\ell^2[(\mathbb{Z}^2)^2]$. This Hamiltonian is associated with a system of two arbitrary particles interacting via a zero-range  and nearest-neighbor potential $\hat v_{\lambda\mu}$ and is expressed as
\begin{equation*}
\hat {\mathbb{H}}_{\lambda\mu}=\hat {\mathbb{H}}_0 + \hat {\mathbb{V}}_{\lambda\mu},\qquad
\lambda, \mu\in\mathbb{R}.
\end{equation*}

The interaction $\hat{\mathbb{V}}_{\lambda \mu}$ is defined as the
multiplication operator by a function $\hat v_{\lambda\mu}(\cdot)$:
\begin{equation*}
(\hat {\mathbb{V}}_{\lambda\mu} \hat\psi)(x,y) = \hat v_{\lambda\mu}(x-y)
\hat\psi(x,y),\,\,\  \hat\psi\in\ell^{2}[(\mathbb{Z}^{2})^{2}],
\end{equation*}
where $\hat v_{\lambda\mu}\in \ell^{1}(\mathbb{Z}^{2}).$

\subsection{The two-particle Hamiltonian: the quasimomentum-space representation}

Let $\mathbb{T}^2= \allowbreak ( \mathbb{R} /2\pi \mathbb{Z})^2  \equiv [-\pi,\pi)^2$ be
the two-dimensional torus, which is the Pontryagin dual group of $\mathbb{Z}^2$
It is equipped with the the Haar measure $\mathrm{d} p$. Let $L^2(\mathbb{T}^2 \times \mathbb{T}^2)$ denote the
Hilbert space of square-integrable functions on
$\mathbb{T}^2\times\mathbb{T}^2.$
Let $\mathcal{F}: \ell^2(\mathbb{Z}^2) \rightarrow L^2(\mathbb{T}^2)$ be the standard Fourier transform. The Fourier transform is defined as:
\begin{equation*}
(\mathcal{F} \hat f)(p)=\frac{1}{(2\pi)} \sum_{x\in\mathbb{Z}^2} \hat f(x) e^{ip\cdot
x},
\end{equation*}
where $p\cdot x: = p_1x_1+p_2x_2$ for $p\in\mathbb{T}^2$ and
$x\in\mathbb{Z}^2.$

In the quasimomentum space, in $L^2(\mathbb{T}^2 \times \mathbb{T}^2)$ the two-particle Hamiltonian $\mathbb{H}_{\lambda\mu}$ is represented as
$$
\mathbb{H}_{\lambda\mu}:=(\mathcal{F}\otimes\mathcal{F}) \hat
{\mathbb{H}}_{\lambda\mu}(\mathcal{F}\otimes\mathcal{F})^*:=\mathbb{H}_0 + \mathbb{V}_{\lambda\mu}.
$$
Here, the free Hamiltonian is given by $ \mathbb{H}_0=(\mathcal{F} \otimes \mathcal{F}) \hat {\mathbb{H}}_0
(\mathcal{F}\otimes \mathcal{F})^*$, which is the multiplication operator:
$$
(\mathbb{H}_0 f)(p,q) = [\epsilon(p) + \gamma\epsilon(q)]f(p,q),\,\,\gamma>0,
$$
where
\begin{equation}\label{U2}
\epsilon(p) = \sum_{s\in\mathbb{Z}^2} \hat \epsilon(s) e^{ip\cdot s},\quad
p\in \mathbb{T}^2,
\end{equation}
is a real-valued function on $\mathbb{T}^2$.

The interaction $\mathbb{V}_{\lambda\mu}$ is a (partial) integral operator defined as $(\mathcal{F} \otimes
\mathcal{F})\hat{\mathbb{V}}_{\lambda\mu} (\mathcal{F}\otimes\mathcal{F})^*$ and it can be written as:
$$
(\mathbb{V}_{\lambda\mu} f)(p,q) = \frac{1}{(2\pi)^2}\int \limits_{\mathbb{T}^2}
v_{\lambda\mu}(p-u) f(u,p+q-u)\mathrm{d} u,
$$
where
\begin{equation}\label{U3}
v_{\lambda\mu}(p)=\sum_{s\in\mathbb{Z}^2} \hat v_{\lambda\mu}(s) e^{ip\cdot s},\quad
p\in \mathbb{T}^2.
\end{equation}

\subsection{The Floquet-Bloch decomposition of $\mathbb{H}_{\lambda\mu}$ and discrete Schr\"odinger operators}\label{subsec:von_neuman}

Since $\hat H_{\lambda\mu}$ commutes with the discrete group $\mathbb{Z}^2$,
which is represented by shift operators on the lattice, we can decompose
the space $L^2(\mathbb{T}^2 \times \mathbb{T}^2)$ and $H_{\lambda \mu}$ into the von Neumann direct integral:
\begin{equation*}
L^{2}(\mathbb{T}^2\times \mathbb{T}^2)\simeq \int\limits_{K\in \mathbb{T}^2} \oplus
L^{2}(\mathbb{T}^2)\,\mathrm{d} K,
\end{equation*}

\begin{equation*}
\mathbb{H}_{\lambda\mu} \simeq \int\limits_{K\in \mathbb{T}^2} \oplus
H_{\lambda\mu}(K)\,\mathrm{d} K.
\end{equation*}
The fiber operator $H_{\lambda \mu}(K)$, $K \in \mathbb{T}^2$, is a self-adjoint operator in $L^2(\mathbb{T}^2)$ and defined as
\begin{equation*}
H_{\lambda \mu}(K) = H_0(K) + V_{\lambda \mu},
\end{equation*}
 where $H_0(K)$ is the unperturbed operator and $V_{\lambda \mu}$ is the perturbation operator.
 The unperturbed operator $H_0$ is the multiplication operator by a function $\mathcal{E}_K$:
 $$
 \mathcal{E}_K(p) = \varepsilon(p) + \gamma \varepsilon(K - p),
 $$
 where $\gamma > 0$. The perturbation operator $V_{\lambda \mu}$ is defined by the integral:
 \begin{align}\label{U4}
 (V_{\lambda\mu} f)(p)= &\frac{1}{(2\pi)^2}\int \limits_{\mathbb{T}^2}v_{\lambda\mu}(p-q)f(q)\mathrm{d} q.
\end{align}

The parameter $K \in \mathbb{T}^2$ is typically referred to as the \emph{two-particle quasi-momentum,} and the fiber $H_{\lambda \mu}(K)$ is known as the \emph{discrete Schr\"odinger operator} associated with the two-particle Hamiltonian $\hat{\mathbb{H}}_{\lambda \mu}$.

\subsection{The essential spectrum of discrete Schr\"odinger operators} \label{subsec:ess_spec}

Since the operator $V_{\lambda \mu}$ is compact, according to Weyl's theorem, for
any $K\in\mathbb{T}^2$ the essential spectrum
$\sigma_{\mathrm{ess}}(H_{\lambda\mu}(K))$ of the operator $H_{\lambda\mu}(K)$ coincides with the spectrum of the operator
$H_0(K),$ i.e.,
\begin{equation}\label{U5}
\sigma_{\mathrm{ess}}(H_{\lambda\mu}(K))=\sigma(H_0(K)).
\end{equation}
In particular, the essential spectrum is given by
\begin{equation*}
\sigma_{\mathrm{ess}}(H_{\lambda\mu}(K))=
[\mathcal{E}_{\min}(K),\mathcal{E}_{\max}(K)],
\end{equation*}
where
\begin{align*}
\mathcal{E}_{\min}(K):= & \min_{p\in  \mathbb{T} ^2}\,\mathcal{E}_K(p), \quad
\mathcal{E}_{\max}(K):= \max_{p\in  \mathbb{T} ^2}\,\mathcal{E}_K(p).
\end{align*}

\section{Main results}\label{sec:main_results}

Let $L^{2,e}(\mathbb{T}^2)$ and $L^{2,o}(\mathbb{T}^2)$ be subspaces of $L^2(\mathbb{T}^2)$. These subspaces consist of even and odd functions, respectively.
More formally, we can write:

\begin{align*}
&L^{2,e}(\mathbb{T}^2)=\{{f}\in L^{2}(\mathbb{T}^2): {f}(p)=
{f}(-p), \mbox{for a.e}\,\,p\in\mathbb{T}^2\}
\end{align*}
and
\begin{align*}
&L^{2,o}(\mathbb{T}^2)=\{{f}\in L^{2}(\mathbb{T}^2): {f}(p)=
-{f}(p), \mbox{for a.e}\,\,p\in\mathbb{T}^2\}.
\end{align*}

\begin{lemma}\label{shuh1}
The following statement is correct:
\begin{equation*}
L^{2}(\mathbb{T}^2)=L^{2,e}(\mathbb{T}^2)\oplus
L^{2,o}(\mathbb{T}^2).
\end{equation*}

\end{lemma}
\begin{proof}
The proof follows from the fact that every element of $L^2(\mathbb{T}^2)$ can be expressed as the sum of two functions, which one of them belongs to $L^{2,e}(\mathbb{T}^2)$ and other one belongs to $L^{2,o}(\mathbb{T}^2)$.

Let $f \in L^{2}(\mathbb{T}^2)$ and we define the even and odd parts of $f$ as follows:
$$
f_e(p)=\frac{f(p)+f(-p)}{2}\in L^{2,e}(\mathbb{T}^2)
$$
and
$$
f_o(p)=\frac{f(p)-f(-p)}{2}\in L^{2,o}(\mathbb{T}^2).
$$
Then, we have
$f(p)=f_e(p)+f_o(p).$

On the other hand, for any function $g(p)$ belonging to the space $L^{2,o}(\mathbb{T}^2)$, we have

$$
\int\limits_{\mathbb{T}^2}g(p)dp=0.
$$

Note that if both $f_e$ and $f_o$ belong to $L^2(\mathbb{T}^2)$, their product $f_e \cdot f_o$ belongs to $L^{2,o}(\mathbb{T}^2)$. Consequently, we have $f_e(p)\bot f_o(p)$ in $L^2(\mathbb{T}^2)$.

\end{proof}

We assume that the potential function $\hat{v}_{\lambda\mu}(\cdot)$  is defined as follows:

\begin{equation}\label{U6}
\hat{v}_{\lambda\mu}(x)= \left\lbrace\begin{array}{ccc}
 \lambda, \quad \quad |x|=0,\\
\frac {\mu}{2},\quad \quad |x|=1,\\
0, \quad \quad |x|>1,
\end{array}\right.
\end{equation}
where $\lambda$ and $\mu$ are real numbers.
Using the relations \eqref{U3}, \eqref{U4} and \eqref{U6}, we can derive the following relations:
\begin{align}\label{U7}
(V_{\lambda\mu}f)(p)=& \frac{\lambda}{(2\pi)^2}\int \limits_{\mathbb{T}^2}f(q)\mathrm{d} q+
\frac{\mu}{(2\pi)^2}\sum\limits_{i=1}^2\cos p_i\int \limits_{\mathbb{T}^2}\cos q_if(q)\mathrm{d} q \\ \nonumber
&+\frac{\mu}{(2\pi)^2}\sum\limits_{i=1}^2\sin p_i\int \limits_{\mathbb{T}^2}\sin q_if(q)\mathrm{d} q ,\quad f\in L^{2}(\mathbb{T}^2).
\end{align}

In addition, we assume that  the function ${\hat{\varepsilon}}(\cdot)$ is defined as follows:
\begin{equation*}
  {\hat{\varepsilon}(x)=}
\left\lbrace\begin{array}{ccc}
2,\quad \quad    |x|=0,\\
-\frac{1}{2},\,\, \quad  |x|=1,\\
0,\quad \quad   |x|>1.
\end{array}\right.
\end{equation*}

Then, using the definition of $\epsilon(p)$ in \eqref{U2}, we get
$$
\epsilon(p) := \sum\limits_{i=1}^2 \big(1-\cos p_i),\quad
p=(p_1,p_2)\in \mathbb{T}^2.
$$

Let $K=0$. The operator $H_0(0)$ is the multiplication operator by the even function $\mathcal{E}_0(p) = (1 + \gamma) \epsilon(p)$ in $L^2(\mathbb{T}^2)$. Therefore, for each $\theta \in \{e, o\}$, the subspace $L^{2,\theta}(\mathbb{T}^2)$ is invariant under $H_0(0)$.
Taking into account that the interaction operator $V_{\lambda \mu}$ has the form given in equation \eqref{U7}, denoting by $V_{\lambda \mu}^e$ and ${V}^o_{\mu}$ the part of  $V_{\lambda \mu}$ acting on $L^{2,e}(\mathbb{T}^2)$  and $L^{2,o}(\mathbb{T}^2)$, respectively, we obtain the following expressions
\begin{align}\label{U8}
({V}_{\lambda\mu}^ef)(p)=\frac{\lambda}{(2\pi)^2}\int \limits_{\mathbb{T}^2}f(q)\mathrm{d} q  +\frac{\mu}{(2\pi)^2}
\sum\limits_{i=1}^2\cos p_i\int \limits_{\mathbb{T}^2}\cos q_if(q)\mathrm{d} q
\end{align}
and
\begin{align}\label{U9}
({V}_{\mu}^o{f})(p)=\frac{\mu}{(2\pi)^2}
\sum\limits_{i=1}^2\sin p_i\int \limits_{\mathbb{T}^2}\sin q_if(q)\mathrm{d} q.
\end{align}
From the above expressions, we can conclude that
$$
{H}_{\lambda\mu}(0)\big|_{L^{2,e}(\mathbb{T}^2)} =
{H}_{\lambda\mu}^e(0):=H_0(0) + {V}_{\lambda\mu}^e
$$
and
$$
{H}_{\lambda\mu}(0)\big|_{L^{2,o}(\mathbb{T}^2)} =
{H}_{\mu}^o(0):=H_0(0) + {V}_{\mu}^o.
$$
Therefore, the spectrum of $H_{\lambda\mu}(0)$, denoted by $\sigma(H_{\lambda\mu}(0))$, can be expressed as the union of the spectra of the restrictions of $H_{\lambda \mu}(0)$ to the subspaces $L^{2,e}(\mathbb{T}^2)$ and $L^{2,o}(\mathbb{T}^2)$:
\begin{equation*}
\sigma\Big({H}_{\lambda\mu}(0)\Big)=
\sigma\Big({H}_{\lambda\mu}(0)\big|_{L^{2,e}(\mathbb{T}^2)}\Big)\cup\sigma\Big({H}_{\lambda\mu}(0)\big|_{L^{2,o}(\mathbb{T}^2)}\Big).
\end{equation*}
Here, $A|_\mathcal{H}$ denotes the restriction of the self-adjoint operator $A$ to the subspace $\mathcal{H}$. Thus, the study of the discrete spectrum of $H_{\lambda\mu}(0)$ can be reduced to studying the spectra of these restrictions.

Our first main result is a generalization of
Theorems 1 and 2 in \cite{ham23}, which we present below.
\begin{theorem}\label{lak1}
Suppose that, for some values of
$\lambda,\mu\in\mathbb{R}$, $H_{\lambda \mu}(0)$ has $n$ eigenvalues lying below
(respectively,  above) the essential spectrum of $H_{\lambda \mu}(0)$.  Then, for each $K \in \mathbb{T}^2$, the operator
$H_{\lambda \mu}(K)$ has at least $n$ eigenvalues lying below
(respectively, above) the essential spectrum of $H_{\lambda \mu}(K)$.
\end{theorem}

Let us define the following nine sets in the $(\lambda, \mu)$ plane:
\begin{equation*}
\begin{aligned}
\mathcal{S}_{0}^+:= & \Big\{(\lambda,\mu)\in\mathbb{R}^2:\,\,\mu>\frac{\pi(1+\gamma)}{8-2\pi}\Big\},\\
\mathcal{S}_{0}:= &\Big\{(\lambda,\mu)\in\mathbb{R}^2:\,\,|\mu|<\frac{\pi(1+\gamma)}{8-2\pi}\Big\},\\
\mathcal{S}_{0}^-:=
&\Big\{(\lambda,\mu)\in\mathbb{R}^2:\,\,\mu<\frac{\pi(1+\gamma)}{2\pi-8}\Big\},\\
\allowbreak
\mathcal{C}_0^+:= &\Big\{(\lambda,\mu)\in\mathbb{R}^2:\,\,2\mu +\lambda-\frac{\lambda\mu}{{1+\gamma}}>0,\,\, \mu<1+\gamma \Big\},\\
\mathcal{C}_1^+:= &\Big\{(\lambda,\mu)\in\mathbb{R}^2:\,\,2\mu +\lambda-\frac{\lambda\mu}{{1+\gamma}}<0\Big\},\\
\mathcal{C}_2^+:= &\Big\{(\lambda,\mu)\in\mathbb{R}^2:\,\,2\mu +\lambda-\frac{\lambda\mu}{{1+\gamma}}>0,\,\, \mu>1+\gamma\Big\}, \\
\mathcal{C}_0^-:= &\Big\{(\lambda,\mu)\in\mathbb{R}^2:\,\,2\mu +\lambda+\frac{\lambda\mu}{{1+\gamma}}>0,\,\, \mu>-(1+\gamma) \Big\},\\
\mathcal{C}_1^-:= &\Big\{(\lambda,\mu)\in\mathbb{R}^2:\,\,2\mu +\lambda+\frac{\lambda\mu}{{1+\gamma}}<0\Big\},\\
\mathcal{C}_2^-:=
&\Big\{(\lambda,\mu)\in\mathbb{R}^2:\,\,2\mu +\lambda+\frac{\lambda\mu}{{1+\gamma}}>0,\,\,
\mu<-(1+\gamma)\Big\},
\end{aligned}
\end{equation*}
where $\gamma>0.$
(see Figures 1).

Furthermore, in the $(\lambda, \mu)$-plane, let us define the following three sets:
\begin{equation*}
\begin{aligned}
\mathcal{D}_{0}^+:= & \Big\{(\lambda,\mu)\in\mathbb{R}^2:\,\,\mu>\frac{\pi(1+\gamma)}{\pi-2}\Big\},\\
\mathcal{D}_{0}:= &\Big\{(\lambda,\mu)\in\mathbb{R}^2:\,\,|\mu|<\frac{\pi(1+\gamma)}{\pi-2}\Big\},\\
\mathcal{D}_{0}^-:=&\Big\{(\lambda,\mu)\in\mathbb{R}^2:\,\,\mu<\frac{\pi(1+\gamma)}{2-\pi}\Big\}.
\end{aligned}
\end{equation*}

\begin{center}
\begin{tabular}{p{2.9in}p{2.9in}}
\includegraphics[scale=0.4]{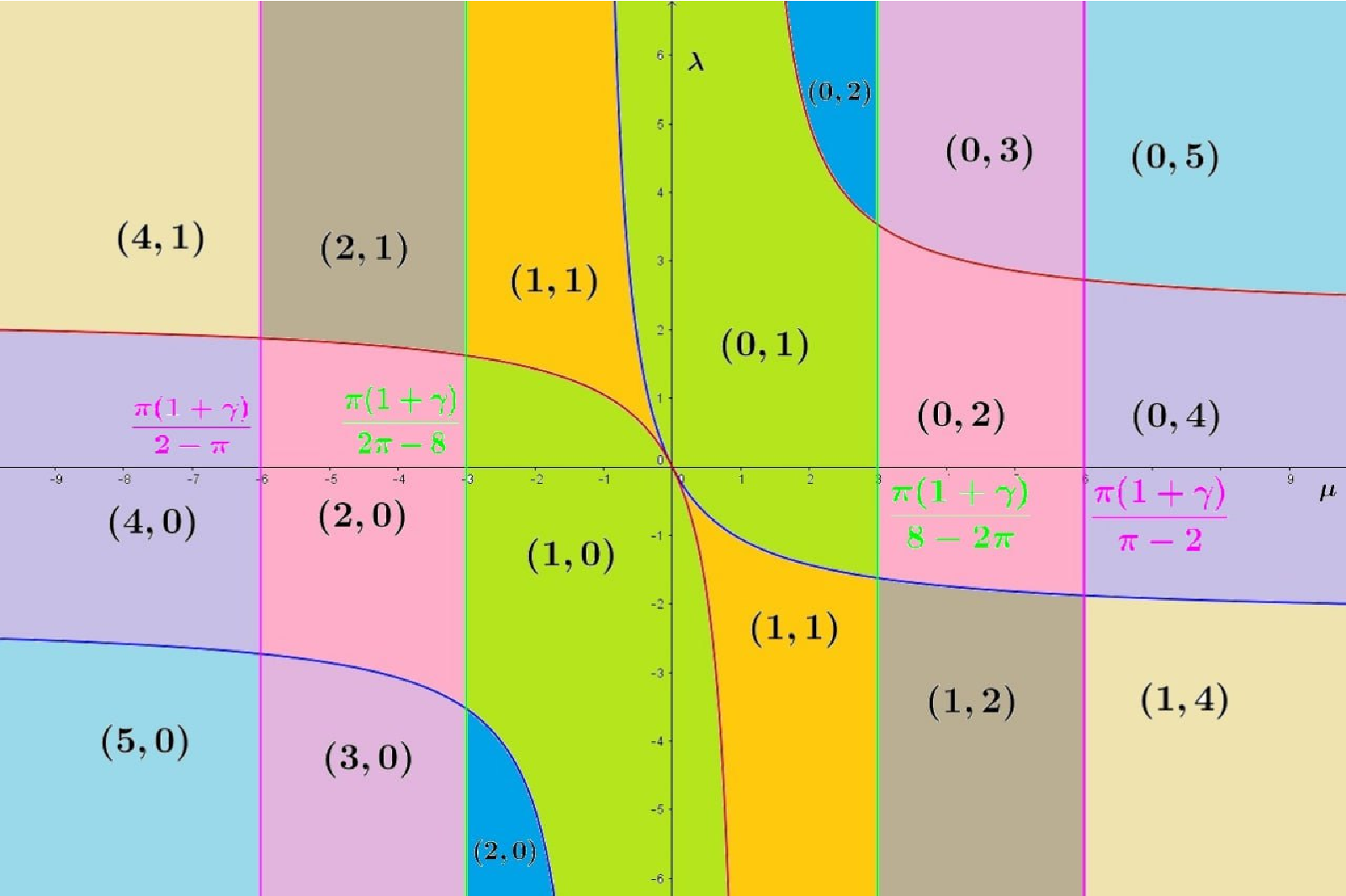}
\end{tabular}\\
Figure 1. A schematic dynamics of the change in the number of eigenvalues depending on $\mu$ and $\lambda$.
\end{center}

It has been found that in each of the above-mentioned components $\mathcal{D}_{0}$, $\mathcal{D}_{0}^\pm$, $\mathcal{S}_{0}$, $\mathcal{S}_{0}^\pm$ and $\mathcal{C}^{\pm}_{k},\,k=0,1,2$, the number of eigenvalues of the operator $H_{\lambda\mu}(0)$ lying outside the essential spectrum \eqref{U5}, remains constant.
These results are established in the next theorem.

\begin{theorem}\label{lak2}
In each of the connected components $\mathcal{D}_{0}$, $\mathcal{D}_{0}^\pm$, $\mathcal{S}_{0}$, $\mathcal{S}_{0}^\pm$ and $\mathcal{C}^{\pm}_{k},\,k=0,1,2$, the number of eigenvalues $n_+(H_{\lambda \mu}(0))$ and $n_-(H_{\lambda \mu}(0))$ of the operator $H_{\lambda \mu}(0)$ lying respectively above and below the essential spectrum of $H_{\lambda \mu}(0)$, remains constant.
\end{theorem}

Theorem \ref{lak2} can be proved in a similar way as it is done in \cite{ham30}.

The following result concerns the number of eigenvalues of the fiber Hamiltonian $H_{\lambda\mu}(K)$ for all $K\in\mathbb{T}^2$ and $(\lambda,\mu)\in\mathbb{R}^2$.

\begin{theorem}\label{lak3}
Let $K\in \mathbb{T}^2$ and $(\lambda, \mu) \in \mathbb{R}^2$. For the numbers
$n_+(H_{\lambda \mu}(K))$ and $n_-(H_{\lambda \mu}(K))$,
which are the number of eigenvalues of the operator $H_{\lambda \mu}(K)$ lying
respectively above and below the essential spectrum $\sigma_{\mathrm{ess}}(H_{\lambda \mu}(K)),$
the following statements hold:

\begin{equation*}
\begin{aligned}
& (\lambda,\mu)\in \mathcal{C}_2^+ \cap \mathcal{S}_{0}^+ \cap  \mathcal{D}_{0}^+& 
\Longrightarrow \qquad  n_+({H}_{\lambda\mu}(K)) =5, \\
& (\lambda,\mu)\in \mathcal{C}_1^+ \cap \mathcal{S}_{0}^+ \cap  \mathcal{D}_{0}^+ & 
\Longrightarrow \qquad  n_+({H}_{\lambda\mu}(K)) \ge4,\\
& (\lambda,\mu)\in (\mathcal{C}_2^+ \cap \mathcal{S}_{0}^+ )\setminus \mathcal{D}_{0}^+& 
\Longrightarrow \qquad  n_+({H}_{\lambda\mu}(K)) \ge3, \\
& (\lambda,\mu)\in \mathcal{C}_2^+ \setminus (\mathcal{S}_{0}^+ \cup \mathcal{D}_{0}^+) \,\mathrm{or}
\,(\lambda,\mu)\in (\mathcal{C}_1^+ \cap \mathcal{S}_{0}^+)\setminus \mathcal{D}_{0}^+ &\hspace*{-14mm}
\Longrightarrow \qquad    n_+({H}_{\lambda\mu}(K)) \ge2,\\
& (\lambda,\mu)\in \mathcal{C}_1^+ \setminus ( \mathcal{S}_{0}^+ \cup \mathcal{D}_{0}^+)& 
\Longrightarrow \qquad  n_+({H}_{\lambda\mu}(K)) \ge1, \\
& (\lambda,\mu)\in \overline{\mathcal{C}_0^+}  &  \hspace*{-8mm}
\Longrightarrow \qquad  n_+({H}_{\lambda\mu}(K))\geq 0,
\end{aligned}
\end{equation*}
and
\begin{equation*}
\begin{aligned}
& (\lambda,\mu)\in \mathcal{C}_2^- \cap \mathcal{S}_{0}^- \cap  \mathcal{D}_{0}^-& 
\Longrightarrow \qquad  n_-({H}_{\lambda\mu}(K)) =5, \\
& (\lambda,\mu)\in \mathcal{C}_1^- \cap \mathcal{S}_{0}^+ \cap  \mathcal{D}_{0}^- & 
\Longrightarrow \qquad  n_-({H}_{\lambda\mu}(K)) \ge4,\\
& (\lambda,\mu)\in (\mathcal{C}_2^- \cap \mathcal{S}_{0}^+ )\setminus \mathcal{D}_{0}^-& 
\Longrightarrow \qquad  n_-({H}_{\lambda\mu}(K)) \ge3, \\
& (\lambda,\mu)\in \mathcal{C}_2^- \setminus (\mathcal{S}_{0}^- \cup \mathcal{D}_{0}^-) \,\mathrm{or}
\,(\lambda,\mu)\in (\mathcal{C}_1^- \cap \mathcal{S}_{0}^-)\setminus \mathcal{D}_{0}^- &\hspace*{-14mm}
\Longrightarrow \qquad    n_-({H}_{\lambda\mu}(K)) \ge2,\\
& (\lambda,\mu)\in \mathcal{C}_1^- \setminus ( \mathcal{S}_{0}^- \cup \mathcal{D}_{0}^-)& 
\Longrightarrow \qquad  n_-({H}_{\lambda\mu}(K)) \ge1, \\
& (\lambda,\mu)\in \overline{\mathcal{C}_0^-}  &  \hspace*{-8mm}
\Longrightarrow \qquad  n_-({H}_{\lambda\mu}(K))\geq 0,
\end{aligned}
\end{equation*}
where  $\overline{\mathcal{A}}$ is the closure of the set
$\mathcal{A}$.
\end{theorem}

The following theorem determines the exact number of eigenvalues of
${H}_{\lambda\mu}^e(0)$ and ${H}_{\lambda\mu}^o(0)$ that lie above
and below the essential spectrum, respectively.

\begin{theorem}\label{lak4}

\begin{itemize}
\item[(i)] If $(\mu, \lambda) \in \overline{\mathcal{S}_0}$, then the operator $H^{e}_{\lambda \mu}(0)$
does not have any eigenvalues outside the essential spectrum.

\item[(ii)] If $(\mu, \lambda) \in \mathcal{S}_0^+$ and $(\mu,\lambda)\in \mathcal{S}_-$,
then the operator $H^e_{\lambda\mu}(0)$ has a unique eigenvalue above and below the essential spectrum, respectively.

\item[(iii)] If $(\mu,\lambda) \in \mathcal{C}_2^+$ and $(\mu,\lambda)\in\mathcal{C}^-_2$,
then the operator $H^e_{\lambda\mu}(0)$ has exactly two eigenvalues lying above and below the essential spectrum, respectively.

\item[(iv)] If $(\mu,\lambda)\in \mathcal{C}_1^+\cup\partial \mathcal{C}_2^+$ and
$(\mu,\lambda)\in \mathcal{C}_1^-\cup\partial \mathcal{C}_2^-,$ then
the operator $H^e_{\lambda\mu}(0)$ has a unique eigenvalue lying above and
below the essential spectrum, respectively.  Here, $\partial A$ denotes the topological boundary of the set $A$.
\end{itemize}
\end{theorem}

\begin{theorem}\label{lak5}
\begin{itemize}
\item[(i)] If $(\mu,\lambda)\in \overline{\mathcal{D}_{0}},$ then the operator $H^o_{\mu}(0)$ does not
have any eigenvalues outside the essential spectrum.

\item[(ii)] If $(\mu,\lambda) \in \mathcal{D}_0^+$, then the operator $H^o_{\mu}(0)$ has a unique
eigenvalue that lies above the essential spectrum with multiplicity two.

\item[(iii)] If $(\mu,\lambda)\in \mathcal{D}_{0}^-,$  then the operator $H^o_{\mu}(0)$
has a unique eigenvalue that lies below the
essential spectrum with multiplicity two.
\end{itemize}
\end{theorem}

\section{Proofs of the main results}\label{sec:proofs}

\subsection{The Lippmann--Schwinger operator}

Let $\{\alpha_i^e,\,\,i=1,2,3\}$ and $\{\beta_j^o,\,\,j=1,2\}$ be systems of vectors in the spaces $L^{2,e}(\mathbb{T}^2)$ and $L^{2,o}(\mathbb{T}^2)$, respectively and these vectors are given by:
\begin{equation}\label{U10}
\alpha_{1}^{e}(p)=\dfrac{1}{2\pi},\,\,\alpha_{2}^{e}(p)=\dfrac{\cos p_1}{\sqrt{2}\pi},\,\,\alpha_{3}^{e}(p)=\dfrac{\cos p_2}{\sqrt{2}\pi},
\end{equation}

\begin{equation}\label{U11}
\beta_{1}^{o}(p)=\dfrac{\sin p_1}{\sqrt{2}\pi},\,\,\beta_{2}^{o}(p)=\dfrac{\sin p_2}{\sqrt{2}\pi}.
\end{equation}
It is easy to verify by inspection that the vectors \eqref{U10} and \eqref{U11} are orthonormal in $L^{2, e}(\mathbb{T}^2)$ and $L^{2, o}(\mathbb{T}^2)$, respectively. For \eqref{U8} and \eqref{U9} using the orthogonal systems \eqref{U10} and \eqref{U11}, respectively, we obtain
\begin{align}\label{U12}
&{V}_{\lambda\mu}^e{f}=\lambda({f},\alpha_{1}^e)\alpha_{1}^e+
\frac{\mu}{2}({f},\alpha_{2}^e)\alpha_{2}^e+
\frac{\mu}{2}({f},\alpha_{3}^e)\alpha_{3}^e
\end{align}
and
\begin{align}\label{U13}
&{V}_{\mu}^o{f}=\frac{\mu}{2}({f},\beta_{1}^o)\beta_{1}^o+
\frac{\mu}{2}({f},\beta_{2}^o)\beta_{2}^o,
\end{align}
where  $(\cdot,\cdot)$ denotes the inner product in $L^{2,\theta}(\mathbb{T}^2),\,\,\theta\in\{e,o\}.$

For any $z \in \mathbb{C} \setminus [\mathcal{E}_{min}(0), \mathcal{E}_{max}(0)]$, we define the Lippmann-Schwinger operator, as described in \cite{ham31}, as follow:
\begin{equation*}
{B}_{\lambda\mu}^{e}(0,z)=-{V}_{\lambda\mu}^{e}{R}_0(0,z)
\end{equation*}
and
\begin{equation*}
{B}_{\mu}^{o}(0,z)=-{V}_{\mu}^{o}{R}_0(0,z).
\end{equation*}
Here, ${R}_0(0,z):= [{H}_0(0)-zI]^{-1}$ is the resolvent operator of the operator ${H}_0(0)$ for
$z\in \mathbb{C} \setminus[\mathcal{E}_{\min}(0),\,\mathcal{E}_{\max}(0)].$

\begin{lemma}\label{shuh2}
For each $\lambda,\mu \in \mathbb{R}$,  the number $z\in\mathbb{C} \setminus[\mathcal{E}_{\min}(0),\,\mathcal{E}_{\max}(0)]$
 is an eigenvalue of the operators
${H}_{\lambda\mu}^{e}(0)$ and ${H}_{\mu}^{o}(0)$ if and only if the number $1$ is an
eigenvalue for
 ${B}_{\lambda\mu}^{e}(0,z)$ and  ${B}_{\mu}^{o}(0,z)$, respectively.
\end{lemma}
The proof of this lemma is straightforward and can be found in various sources, such as \cite{ham32}.
Therefore, we do not provide a proof in this paper.

The expressions \eqref{U12} and \eqref{U13} lead to the following algebraic linear systems depending on
$x_i:=({\varphi},\alpha_{i}^e),\,\,i=1,2,3$ and $y_j:=({\varphi},\beta_{j}^o),\,\,j=1,2$:
\begin{equation*}
\left\lbrace\begin{array}{ccc}
(1+\lambda a(z))x_{1}+ \mu b(z) x_{2}+\mu  b(z)x_{3}=0,\\
\lambda b(z)x_{1}+(1+\mu c(z))x_{2}+\mu  e(z)x_{3}=0,\\
\lambda b(z)x_{1}+\mu e(z)x_{2}+(1+\mu c(z))x_{3}=0
\end{array}\right.
\end{equation*}
and
\begin{equation*}
\left\lbrace\begin{array}{ccc}
(1+\mu f(z))y_{1}=0,\\
(1+\mu f(z))y_{2}=0,
\end{array}\right.
\end{equation*}
which are equivalent to the following Lippmann-Schwinger equations, respectively
\begin{align*}
{B}_{\lambda\mu}^{e}(0,z){\varphi}={\varphi}, \,\,\,{\varphi}
\in L^{2,e}(\mathbb{T}^2)
\end{align*}
and
\begin{align*}
{B}_{\mu}^{o}(0,z){\varphi}={\varphi}, \,\,\,{\varphi}
\in L^{2,o}(\mathbb{T}^2).
\end{align*}
Here
\begin{align}
\label{U14} &a(z) =\frac{1}{(2\pi)^2}\int\limits_
{\mathbb{T}^2}\frac{\mathrm{d} p}{\mathcal{E}_0(p)-z},\\
\label{U15} &b(z) =\frac{1}{(2\pi)^2}\int\limits_
{\mathbb{T}^2}\frac{\cos p_i\,\mathrm{d} p}{\mathcal{E}_0(p)-z},\\
\label{U16} &c(z) =\frac{1}{(2\pi)^2}\int\limits_
{\mathbb{T}^2}\frac{\cos^2 p_i\,\mathrm{d} p}{\mathcal{E}_0(p)-z},\\
\label{U17} &e(z) =\frac{1}{(2\pi)^2}\int\limits_
{\mathbb{T}^2}\frac{\cos p_i\cos p_j\,\mathrm{d} p}{\mathcal{E}_0(p)-z},\\
\label{U18} &f(z) =\frac{1}{(2\pi)^2}\int\limits_
{\mathbb{T}^2}\frac{\sin^2 p_i\,\mathrm{d} p}{\mathcal{E}_0(p)-z},
\end{align}
$i,j=1,2,\,\, i\neq j.$

If we use the formulas \eqref{U14}--\eqref{U18}, we can write for $z\in\mathbb{C} \setminus[\mathcal{E}_{\min}(0),\,\mathcal{E}_{\max}(0)]$ the determinants of the operators
 $I-{B}_{\lambda\mu}^{e}(0,z)$ and $I-{B}_{\mu}^{o}(0,z)$  as follow
$$
\Delta^e_{\lambda\mu}(z):=\det[I-{B}_{\lambda\mu}^{e}(0,z)]
$$
and
$$
\Delta^o_{\mu}(z):=\det[I-{B}_{\mu}^{o}(0,z)].
$$

\begin{lemma}\label{shuh3}
A number $z\in\mathbb{C} \setminus[\mathcal{E}_{\min}(0),\,\mathcal{E}_{\max}(0)]$ is an eigenvalue of the
operators ${H}^e_{\lambda\mu}(0)$ and ${H}^o_{\mu}(0)$ if and only if
\begin{equation*}\label{cont01}
\Delta^e_{\lambda\mu}(z)=0 \,\, \text{and}\,\, \Delta^o_{\mu}(z)=0,
\end{equation*}
respectively.
\end{lemma}
The proof of this lemma follows a standard approach, which is similar to that presented in the references \cite{ham33} and \cite{ham34}.

\begin{lemma}\label{shuh4}
For any $\lambda, \mu \in \mathbb{R}$, the determinants $\Delta^e_{\lambda\mu}(z)$ and $\Delta^o_{\mu}(z)$
are defined as follow:
\begin{align*}
&\Delta^e_{\lambda\mu}(z) :=\widetilde{\Delta}^e_{\lambda\mu}(z)\Delta^o_{\mu}(z),
\end{align*}
where
\begin{align*}
&\widetilde{\Delta}^e_{\lambda\mu}(z)=[(1 +\lambda a(z))(1 +\mu (c(z)+e(z)))-2\lambda\mu b^2(z)],
\end{align*}
\begin{align}\label{U19}
&\Delta^e_{\mu}(z)=[1+\mu(c(z)-e(z))]
\end{align}
and
\begin{align}\label{U20}
&\Delta^o_{\mu}(z) := [1+\mu f(z)]^2.
\end{align}
\end{lemma}
\begin{proof}
Direct calculation of the determinant yields the result.
\end{proof}

\begin{lemma}\label{shuh5}
The functions $a(z)$, $b(z)$, $c(z)$, $e(z)$ and $f(z)$, which are defined in $\mathbb{C} \setminus [\mathcal{E}_{\min}(0), \mathcal{E}_{\max}(0)]$, are strictly increasing, real-valued and positive in the interval $(-\infty, {\mathcal{E}}_{\min }(0))$. They are also strictly increasing and negative in the interval $({\mathcal{E}}_{\max}(0),+\infty)$. Moreover, they have the following asymptotic behavior:
\begin{align*}
&a(z)=-\frac{1}{2\pi{(1+\gamma)}}\ln(\mathcal{E}_{\min}(0)-z)+\frac{5\ln 2}{2\pi{(1+\gamma)}}+o(1),as\,\,z\nearrow
\mathcal{E}_{\min}(0),\\
&b(z)=-\frac{1}{2\pi{(1+\gamma)}}\ln(\mathcal{E}_{\min}(0)-z)+\frac{5\ln 2-\pi}{2\pi{(1+\gamma)}}+
o(1),
as\,\,z\nearrow \mathcal{E}_{\min }(0),\\
&c(z)=-\frac{1}{2\pi{(1+\gamma)}}\ln(\mathcal{E}_{\min}(0)-z)+\frac{5\ln 2-3\pi+8}{2\pi{(1+\gamma)}}+o(1),
as\,\,z\nearrow \mathcal{E}_{\min}(0),\\
&e(z)=-\frac{1}{2\pi{(1+\gamma)}}\ln(\mathcal{E}_{\min}(0)-z)+\frac{5\ln 2+\pi-8}{2\pi{(1+\gamma)}}+o(1),
as\,\,z\nearrow \mathcal{E}_{\min}(0),\\
&f(z)=\frac{\pi-2}{\pi{(1+\gamma)}}, \quad as\,\,z\nearrow
\mathcal{E}_{\min}(0)
\end{align*}
and
\begin{align*}
&a(z)=\frac{1}{2\pi{(1+\gamma)}}\ln(z-\mathcal{E}_{\max}(0))-\frac{5\ln 2}{2\pi{(1+\gamma)}}+o(1),as\,\,z\searrow
\mathcal{E}_{\max}(0),\\
&b(z)=\frac{1}{2\pi{(1+\gamma)}}\ln(z-\mathcal{E}_{\max}(0))-\frac{5\ln 2-\pi}{2\pi{(1+\gamma)}}+o(1),
as\,\,z\searrow \mathcal{E}_{\max }(0),\\
&c(z)=\frac{1}{2\pi{(1+\gamma)}}\ln(z-\mathcal{E}_{\max}(0))-\frac{5\ln 2-3\pi+8}{2\pi{(1+\gamma)}}+o(1),
as\,\,z\searrow \mathcal{E}_{\max }(0),\\
&e(z)=\frac{1}{2\pi{(1+\gamma)}}\ln(z-\mathcal{E}_{\max}(0))-\frac{5\ln 2+\pi-8}{2\pi{(1+\gamma)}}+o(1),
as\,\,z\searrow \mathcal{E}_{\max }(0),\\
&f(z)=\frac{2-\pi}{\pi{(1+\gamma)}}, \quad as\,\,z\searrow
\mathcal{E}_{\max}(0),
\end{align*}
where $\gamma>0.$
\end{lemma}
The lemma \ref{shuh5} can be proved using a similar approach to that of proposition 4.4 in \cite{ham34}.

\begin{lemma}\label{shuh6}
The function $\widetilde{\Delta}^e_{\lambda\mu}(z)$ is real-valued for all $z \in \mathbb{C} \setminus [\mathcal{E}_{min}(0), \mathcal{E}_{max}(0)]$, and it has the following asymptotic behavior:
\begin{itemize}
\item[(\rm{i})]
\begin{align*}
&\widetilde{\Delta}^e_{\lambda\mu}(z)=-\frac{1}{2\pi{(1+\gamma)}}S^{-}(\lambda,\mu;\gamma)
\ln(\mathcal{E}_{\min}(0)-z)+B^{-}(\lambda,\mu;\gamma)+o(1), as
\,\,z\nearrow \mathcal{E}_{\min}(0),
\end{align*}
where
\begin{align*}
&S^{-}(\lambda,\mu;\gamma)=
2\mu +\lambda+\frac{\lambda\mu}{{1+\gamma}},\\
&B^{-}(\lambda,\mu;\gamma)=1+\frac{5\ln 2-\pi}{\pi{(1+\gamma)}}\mu+
\frac{5\ln 2}{2\pi{(1+\gamma)}}\lambda+\frac{5\ln 2-\pi}{\pi{(1+\gamma)}}\Big(\frac{5\ln 2}{2\pi{(1+\gamma)}}-1\Big)\lambda\mu.
\end{align*}

\item[(\rm{ii})]
\begin{align*}
&\widetilde{\Delta}^e_{\lambda\mu}(z)=\frac{1}{2\pi{(1+\gamma)}}S^{+}(\lambda,\mu;\gamma)
\ln(z-\mathcal{E}_{\max}(0))+B^{+}(\lambda,\mu;\gamma)+o(1), as
\,\,z\searrow \mathcal{E}_{\max}(0),
\end{align*}
where
\begin{align*}
&S^{+}(\lambda,\mu;\gamma)=
2\mu +\lambda-\frac{\lambda\mu}{{1+\gamma}},\\
&B^{+}(\lambda,\mu;\gamma)=1-\frac{5\ln 2-\pi}{\pi{(1+\gamma)}}\mu-
\frac{5\ln 2}{2\pi{(1+\gamma)}}\lambda+\frac{5\ln 2-\pi}{\pi{(1+\gamma)}}\Big(\frac{5\ln 2}{2\pi{(1+\gamma)}}-1\Big)\lambda\mu.
\end{align*}
\end{itemize}
\end{lemma}
By using Lemma \ref{shuh5}, we can prove Lemma \ref{shuh6}.

\begin{lemma}\label{shuh7}
The function $\Delta^e_\mu(z)$ is a real-valued function in $\mathbb{C}\setminus[\mathcal{E}_{\min}(0),\mathcal{E}_{\max}(0)]$ and it has the following asymptotic behavior:
\begin{align*}
&\lim\limits_{z\rightarrow {\pm\infty}}\Delta^e_{ \mu }(z)=1,\\
&\lim\limits_{z \nearrow \mathcal{E}_{\min}(0)} \Delta^e_{\mu }(z) =1+\frac{8-2\pi}{\pi{(1+\gamma)}}\mu,\\
&\lim\limits_{z \searrow \mathcal{E}_{\max}(0)} \Delta^e_{\mu }(z) =1-\frac{8-2\pi}{\pi{(1+\gamma)}}\mu.
\end{align*}

\end{lemma}
\begin{proof}
The first equality follows from the Lebesgue's dominated convergence theorem. The proofs of the last two equalities can be derived using Lemma \ref{shuh5}.
\end{proof}

Next, we will study the number and locations of the roots of the function $\Delta^e_\mu$, which is defined in equation \eqref{U19}.

\begin{lemma}\label{shuh8} Let $\mu\in\mathbb{R}.$
\begin{itemize}
\item[(\rm{i})]If $\mu<\frac{\pi(1+\gamma)}{2\pi-8},$ then the function $\Delta^{e}_{\mu}(\cdot)$
has a unique zero $\zeta^-(\mu)$, which is located in the
interval $(-\infty, \mathcal{E}_{\min}(0))$ and it has no zeros in the interval $(\mathcal{E}_{\max}(0),+\infty)$.

\item[(\rm{ii})] If $\mu\in[\frac{\pi(1+\gamma)}{2\pi-8},
\frac{\pi(1+\gamma)}{8-2\pi}],$ then the function $\Delta^{e}_{\mu}(\cdot)$  has no zeros
in $\mathbb{C}\setminus [\mathcal{E}_{\min}(0),\mathcal{E}_{\max}(0)]$.

\item[(\rm{iii})] If $\mu>\frac{\pi(1+\gamma)}{8-2\pi},$
then the function $\Delta^{e}_{\mu}(\cdot)$  has a unique zero
$\zeta^+(\mu)$, which is located in the
interval $(\mathcal{E}_{\max}(0),+\infty)$ and it has no zeros in the interval $(-\infty, \mathcal{E}_{\min}(0))$.
\end{itemize}
\end{lemma}
The lemma \ref{shuh8} can be proved in a similar way to lemma 4.5 in \cite{ham34}.

\begin{lemma}\label{shuh9}
The function $\Delta^o_\mu(z)$ is real-valued for all $z\in\mathbb{C}\setminus[\mathcal{E}_{\min}(0),\mathcal{E}_{\max}(0)]$, and it has the following asymptotic behavior:
\begin{align*}
&\lim\limits_{z\rightarrow {\pm\infty}}\Delta^o_{ \mu }(z)=1,\\
&\lim\limits_{z \nearrow \mathcal{E}_{\min}(0)} \Delta^o_{\mu }(z) =\Big(1+\frac{\pi-2}{\pi{(1+\gamma)}}\mu\Big)^2,  \\
&\lim\limits_{z \searrow \mathcal{E}_{\max}(0)} \Delta^o_{\mu }(z) =\Big(1-\frac{\pi-2}{\pi{(1+\gamma)}}\mu\Big)^2.
\end{align*}

\end{lemma}
\begin{proof}
The first  item follows directly from the Lebesgue dominated convergence theorem. The proofs for the last two  items can be obtained using Lemma \ref{shuh5}.
\end{proof}

Now, we will study the number and location of the roots of the function $\Delta^{o}_{\mu}$,
 which is defined by equation \eqref{U20}.

\begin{lemma}\label{shuh10} Let $\mu\in\mathbb{R}.$
\begin{itemize}
\item[(\rm{i})]If
$\mu<\frac{\pi(1+\gamma)}{2-\pi},$ then the function $\Delta^{o}_{\mu}(\cdot)$ has a
unique zero, which is located in the
interval $(-\infty,\mathcal{E}_{\min}(0))$ and it has no zeros in the
interval $(\mathcal{E}_{\max}(0),+\infty)$.

\item[(\rm{ii})] If $\mu\in[\frac{\pi(1+\gamma)}{2-\pi},
\frac{\pi(1+\gamma)}{\pi-2}],$ then the function $\Delta^{o}_{\mu}(\cdot)$  has no zeros
in $\mathbb{C}\setminus [\mathcal{E}_{\min}(0),\mathcal{E}_{\max}(0)]$.

\item[(\rm{iii})] If $\mu>\frac{\pi(1+\gamma)}{\pi-2},$
then the function $\Delta^{o}_{\mu}(\cdot)$  has a unique zero
 in the
interval $(\mathcal{E}_{\max}(0),+\infty)$, and it has no zeros in
the
interval $(-\infty, \mathcal{E}_{\min}(0))$.
\end{itemize}
\end{lemma}
The lemma \ref{shuh10} can be proven in a similar way as lemma 4.5 in \cite{ham34}.

The following lemma provides the dependence to the number of zeros
of $\widetilde{\Delta}^{e}_{\lambda\mu}$ in $(\mathcal{E}_{\max}(0),+\infty)$ on the parameters $\mu$ and $\lambda$.

\begin{lemma}\label{shuh11}
Let $(\mu,\lambda)\in\mathbb{R}^2.$

\begin{itemize}
\item[(a)] If $2\mu +\lambda-\frac{\lambda\mu}{{1+\gamma}}\ge0$ and $\mu<1+\gamma,$ then the function    $\widetilde{\Delta}^{e}_{\lambda\mu}(\cdot)$ has no zeros in the interval $(\mathcal{E}_{\max}(0),+\infty).$

\item[(b)] If $2\mu +\lambda-\frac{\lambda\mu}{{1+\gamma}}<0$ or $2\mu +\lambda-\frac{\lambda\mu}{{1+\gamma}}=0$ with $\mu>1+\gamma,$  then the function $\widetilde{\Delta}^{e}_{\lambda\mu}(\cdot)$ has a unique zero in the interval $(\mathcal{E}_{\max}(0),+\infty).$

\item[(c)] If $2\mu +\lambda-\frac{\lambda\mu}{{1+\gamma}}>0$ and $\mu>1+\gamma,$ then the function $\widetilde{\Delta}^{e}_{\lambda\mu}(\cdot)$ has two zeros in the interval $(\mathcal{E}_{\max}(0),+\infty).$
\end{itemize}

\end{lemma}

The zeros of $\widetilde{\Delta}^{e}_{\lambda\mu}$ in the interval $(-\infty, \mathcal{E}_{min}(0))$ are investigated in the following lemma.

\begin{lemma}\label{shuh12}
Let $(\mu,\lambda)\in\mathbb{R}^2.$

\begin{itemize}
\item[(a)] If $2\mu +\lambda+\frac{\lambda\mu}{{1+\gamma}}\ge0$ and $\mu>-(1+\gamma),$ then the function $\widetilde{\Delta}^{e}_{\lambda\mu}(\cdot)$ has no zeros in the interval $(-\infty,\mathcal{E}_{\min}(0)).$

\item[(b)] If $2\mu +\lambda+\frac{\lambda\mu}{{1+\gamma}}<0$ or $2\mu +\lambda+\frac{\lambda\mu}{{1+\gamma}}=0$ with $\mu<-(1+\gamma),$  then the function $\widetilde{\Delta}^{e}_{\lambda\mu}(\cdot)$ has a unique zero in the interval $(-\infty,\mathcal{E}_{\min}(0)).$

\item[(c)] If $2\mu +\lambda+\frac{\lambda\mu}{{1+\gamma}}>0$ and $\mu<-(1+\gamma),$ then the function $\widetilde{\Delta}^{e}_{\lambda\mu}(\cdot)$ has two zeros in the interval $(-\infty,\mathcal{E}_{\min}(0)).$
\end{itemize}

\end{lemma}
The lemmas \ref{shuh11} and \ref{shuh12} can be proved in a similar manner to lemma 4.6 in \cite{ham34}.

\noindent{\textit{Proof of Theorem \ref{lak4}.}}
The proof of the assertions of the theorem are based on the lemmas \ref{shuh3}, \ref{shuh7}, \ref{shuh8}, and \ref{shuh11}, as well as \ref{shuh12}.

\noindent{\textit{Proof of Theorem \ref{lak5}.}}
The proof of the assertions of the theorem are based on the lemmas \ref{shuh3}, \ref{shuh9} and \ref{shuh10}.

\subsection{The discrete spectrum of ${H}_{\lambda\mu}(K)$}

For every $n\ge1$, we define
\begin{equation*}
e_n(K;\lambda,\mu):= \sup\limits_{\phi_1,\ldots,\phi_{n-1}\in
L^{2,e}(\mathbb{T}^2)}\,\,\inf\limits_{{\psi}
\in[\phi_1,\ldots,\phi_{n-1}]^\perp,\,\|{\psi}\|=1}
({H}_{\lambda\mu}(K){\psi},{\psi})
\end{equation*}
and
\begin{equation*}
E_n(K; \lambda,\mu):= \inf\limits_{\phi_1,\ldots,\phi_{n-1}\in
L^{2,e}(\mathbb{T}^2)}\,\,\sup\limits_{{\psi}
\in[\phi_1,\ldots,\phi_{n-1}]^\perp,\,\|{\psi}\|=1}
({H}_{\lambda\mu}(K){\psi},{\psi}).
\end{equation*}
By the minimax principle, $e_n(K;\lambda,\mu)\le \mathcal{E}_{\min}(K)$ and
$E_n(K;\lambda,\mu)\ge \mathcal{E}_{\max}(K).$ Since, the rank of
$V_{\lambda\mu}$ does not exceed three, by choosing suitable elements
$\phi_1$, $\phi_2$ and $\phi_3$ from
the range of $V_{\lambda\mu}$  one concludes that
$e_n(K;\lambda,\mu) = \mathcal{E}_{\min}(K)$ and $E_n(K;\lambda,\mu) =
\mathcal{E}_{\max}(K)$ for all $n\ge4.$

\begin{lemma}\label{shuh13}
Let $n\ge1$ and $i\in\{1,2\}.$  For every fixed $K_j\in\mathbb{T},$
$j\in\{1,2\}\setminus\{i\},$  the map
$$
K_i\in\mathbb{T} \mapsto \mathcal{E}_{\min}((K_1,K_2)) - e_n((K_1,K_2);\lambda,\mu)
$$
is non-increasing in $(-\pi,0]$ and non-decreasing in $[0,\pi]$.
Similarly, for every fixed $K_j\in\mathbb{T},$ $j\in\{1,2\}\setminus\{i\},$
the map
$$
K_i\in\mathbb{T} \mapsto E_n((K_1,K_2);\lambda,\mu) - \mathcal{E}_{\max}((K_1,K_2))
$$
is non-increasing in $(-\pi,0]$ and non-decreasing in $[0,\pi]$.
\end{lemma}

\begin{proof}
Without loss of generality we can assume that $i=1.$ For given ${\psi}\in
L^{2,e}(\mathbb{T}^2)$ we consider
$$
(({H}_0(K) - \mathcal{E}_{\min}(K)){\psi},{\psi})=(1+\gamma)\int \limits_{\mathbb{T}^2}
\sum\limits_{i=1}^2 \cos\tfrac{K_i}{2}\,\big(1-\cos
q_i\big)|\psi(q)|^2\,\mathrm{d} q, \quad K:=(K_1,K_2).
$$
Clearly, the map $K_1\in\mathbb{T}\mapsto (({H}_0(K) -
\mathcal{E}_{\min}(K)){\psi},{\psi})$ is non-increasing  in $(-\pi,0]$ and is
non-decreasing in $[0,\pi].$ Since ${V}_{\lambda\mu}$ is independent
of $K,$ by definition of $e_n(K;\lambda,\mu)$ the map
$K_1\in\mathbb{T}\mapsto e_n(K;\lambda,\mu) - \mathcal{E}_{\min}(K)$ is
non-increasing in $(-\pi,0]$ and is non-decreasing   in $[0,\pi].$

The same argument holds for $K_i\mapsto E_n(K;\lambda,\mu) - \mathcal{E}_{\max}(K)$.
\end{proof}

\noindent{\textit{Proof of Theorem \ref{lak1}.}} By using Lemma
\ref{shuh13} for any $K\in\mathbb{T}^2$ and $m\ge1$ we have
\begin{equation}\label{U21}
0\le \mathcal{E}_{\min}(0) - e_m(0;\lambda,\mu) \le \mathcal{E}_{\min}(K) -
e_m(K;\lambda,\mu)
\end{equation}
and
\begin{equation*}
E_m(K;\lambda,\mu) - \mathcal{E}_{\max}(K) \ge E_m(0;\lambda,\mu) -
\mathcal{E}_{\max}(0) \ge 0.
\end{equation*}
By the assumption of the theorem \ref{lak1}, $ e_n(0;\lambda,\mu)$ is a discrete eigenvalue of
${H}_{\lambda\mu}(0)$ for some $\lambda,\mu\in\mathbb{R}.$ Therefore,
$\mathcal{E}_{\min}(0) - e_n(0;\lambda,\mu)>0,$ and hence, according to
\eqref{U21} and \eqref{U5}
$e_n(K;\lambda,\mu)$ is a discrete eigenvalue of
${H}_{\lambda\mu}(K)$ for any $K\in\mathbb{T}^2.$ Since
$e_1(K;\lambda,\mu)\le \ldots \le e_n(K;\lambda,\mu)<\mathcal{E}_{\min}(K),$
it follows that ${H}_{\lambda\mu}(K)$ has at least $n$ eigenvalues
below its essential spectrum. The same argument applies to $E_n(K;\lambda,\mu)$.\,\,\hfill $\qed$
\smallskip

\noindent{\textit{Proof of Theorem \ref{lak3}}} can be
obtained by combining the results of Theorem \ref{lak1} with Theorem
\ref{lak4} and Theorem \ref{lak5}.\,\hfill $\square$


\begin{thebibliography}{99}

\bibitem{ham1}
D. Mattis, \textquotedblleft {The few-body problem on a lattice}\textquotedblright, Rev. Mod. Phys. {\bf 58}, 361--379, (1986).


\bibitem{ham2}
A.M.Khalkhuzhaev, J.I.Abdullaev, J.Kh.Boymurodov, \textquotedblleft The number of eigenvalues of the three-particle Schr\"odinger operator on three dimensional lattice\textquotedblright, {Lobachevskii J. Math.},  {\bf 43}:12, 3486--3495, (2022).

\bibitem{ham3}
A.M. Khalkhuzhaev, Sh.I. Khamidov, H.Sh. Mahmudov, \textquotedblleft On the Existence of Eigenvalues of the One Particle Discrete Schr\"{o}dinger Operators\textquotedblright,  {AIP Conf. Proc.} 3004, 020007, (2024).

\bibitem{ham4}
I.N.Bozorov, Sh.I.Khamidov, S.N.Lakaev, \textquotedblleft The number and location of eigenvalues of the two particle discrete Schr\"odinger operators\textquotedblright, Lobachevskii J. Math. \textbf{43}:11, 47--58, (2022).

\bibitem{ham5}
S. Albeverio, S.N. Lakaev, A.M. Khalkhujaev,  \textquotedblleft Number of Eigenvalues
of the Three-Particle Schr{\"o}dinger Operators on Lattices\textquotedblright, Markov
Process. Relat. Fields. {\bf18}, 387--420, (2012).

\bibitem{ham6}
S. Albeverio, S.N. Lakaev, Z.I. Muminov, \textquotedblleft Schr\"{o}dinger
operators on lattices. The Efimov effect and discrete spectrum
asymptotics\textquotedblright,  Ann. Henri Poincar\'{e}. {\bf 5}, 743--772,  (2004).

\bibitem{ham7}
  {S.N. Lakaev}, \textquotedblleft {The Efimov's effect
of the three identical quantum particle on a lattice}\textquotedblright, Funct. Anal.
Appl. {\bf 27}, 15--28, (1993).

\bibitem{ham8}
S.N. Lakaev, E. \"Ozdemir, \textquotedblleft {The existence and
location of eigenvalues of the one particle Hamiltonians on
lattices}\textquotedblright, Hacettepe J. Math. Stat. {\bf45}, 1693--1703, (2016).


\bibitem{ham9}
S.N.Lakaev, Sh.I.Khamidov, \textquotedblleft On the number and location of eigenvalues of the two particle Schr\"odinger operator on a lattice\textquotedblright, Lobachevskii J. Math. \textbf{43}:12, 135--145, (2022).


\bibitem{ham10}
S.S. Lakaev, \textquotedblleft Infiniteness of the Discrete Spectrum of Two-Particle Discrete Schr\"{o}dinger Operators\textquotedblright,  Lobachevskii J. Math. \textbf{44}:7, 2781--2789, (2023).

\bibitem{ham11}
S.S. Lakaev, G.I. Ismoilov, O.I. Kurbonov,  \textquotedblleft The Spectrum of a Non-local Discrete Schr\"{o}dinger Operator with a Delta Potential on the One-Dimentional Lattice\textquotedblright,  Lobachevskii J. Math. \textbf{44}:2, 607--613, (2023).

\bibitem{ham111}
S.S. Lakaev, O.I. Kurbanov, V.U. Aktamova,  \textquotedblleft Threshold Analysis of the One-Rank Perturbation Non-Local Discrete Laplacian\textquotedblright,  Lobachevskii J. Math. \textbf{43}:8, 2187--2193, (2022).


\bibitem{ham12}
V. Bach, W.~de Siqueira Pedra, S.N. Lakaev,  \textquotedblleft Bounds on the discrete
spectrum of lattice Schr{\"o}dinger operators\textquotedblright, J. Math. Phys. {\bf
59}:2, 022109, (2017).


\bibitem{ham13}
L.D. Faddeev, S.P. Merkuriev, \textquotedblleft Quantum Scattering Theory for Several
Particle Systems\textquotedblright, (Doderecht: Kluwer Academic Publishers, 1993).

\bibitem{ham14}
V.N. Efimov, \textquotedblleft Weakly-bound states of three resonantly-interacting particles\textquotedblright, Sov. J. Nucl. Phys. \textbf{12}:5, 589--595, (1971).

\bibitem{ham141}
S.N.Lakaev, A.T.Boltaev, \textquotedblleft The Essential Spectrum of a Three Particle Schr\"odinger Operator on Lattices\textquotedblright, Lobachevskii J. Math. \textbf{44}:3, 1176--1187, (2023).


\bibitem{ham15}
G. Dell'Antonio, Z.I. Muminov,  Y.M. Shermatova, \textquotedblleft On the number of
eigenvalues of a model operator related to a system of three
particles on lattices\textquotedblright, J. Phys. A \textbf{44}, 315302, (2011).


\bibitem{ham151}
Z.E. Muminov, S.S. Lakaev, N.M. Aliev,  \textquotedblleft On the Essential Spectrum of Three-Particle Discrete Schr\"odinger Operators with Short-Range Potentials\textquotedblright,  Lobachevskii J. Math. \textbf{42}:6, 1304--1316, (2021).



\bibitem{ham16}
{I. Bloch,}  \textquotedblleft Ultracold quantum gases in optical lattices\textquotedblright, Nat. Phys.
{\bf1}, 23--30,  (2005).

\bibitem{ham17}
K. Winkler, G. Thalhammer, F. Lang, R. Grimm, J. Hecker Denschlag,
A.J. Daley, A. Kantian, H.P. B\"uchler, P. Zoller, \textquotedblleft {Repulsively bound atom
pairs in an optical lattice}\textquotedblright, Nature {\bf 441}, 853--856, (2006).

\bibitem{ham18}
D. Jaksch, C. Bruder, J. Cirac, C.W. Gardiner, P. Zoller, \textquotedblleft {Cold bosonic
atoms in optical lattices}\textquotedblright, Phys. Rev. Lett. {\bf 81}, 3108--3111, (1998).

\bibitem{ham19}
D. Jaksch, P. Zoller, \textquotedblleft {The cold atom
Hubbard toolbox}\textquotedblright, Ann. Phys. {\bf315}, 52--79, (2005).

\bibitem{ham20}
M. Lewenstein, A. Sanpera, V. Ahufinger, \textquotedblleft {Ultracold Atoms
in Optical Lattices: Simulating Quantum Many-body Systems}\textquotedblright, Oxford
University Press, Oxford, (2012).

\bibitem{ham21}
C. Ospelkaus, S. Ospelkaus, L. Humbert, P. Ernst, K. Sengstock, K.
Bongs, \textquotedblleft {Ultracold heteronuclear molecules in a 3d optical lattice}\textquotedblright, Phys.Rev. Lett.
{\bf 97}, (2006).

\bibitem{ham22}
A.K.\,Motovilov, W.\,Sandhas, and V.B.\,Belyaev, \textquotedblleft {Perturbation of a lattice
spectral band by a nearby resonance}\textquotedblright, J. Math. Phys. \textbf{42}, 2490--2506,
(2001).

\bibitem{ham23}
{S.~Albeverio, S.N.~Lakaev, K.A.~Makarov, Z.I.~Muminov,} \textquotedblleft{The
Threshold Effects for the Two-particle Hamiltonians on Lattices\textquotedblright,}
Comm. Math. Phys. {\bf 262}, 91--115, (2006).

\bibitem{ham24}
A. Mogilner, \textquotedblleft Hamiltonians in solid-state physics as multiparticle
discrete Schr\"odinger operators: Problems and results\textquotedblright, Advances in
Societ Math. {\bf 5}, 139--194, (1991).

\bibitem{ham25}
P.A. Faria Da Veiga, L. Ioriatti, M. O'Carroll, \textquotedblleft {Energy-momentum
spectrum of some two-particle lattice Schr\"odinger Hamiltonians}\textquotedblright,
Phys. Rev. E  {\bf 66}, 016130, (2002).


\bibitem{ham26}
A.V. Sobolev, \textquotedblleft {The Efimov effect. Discrete spectrum asymptotics}\textquotedblright, Commun. Math. Phys.
\textbf{156}:1, 101--126, (1993).

\bibitem{ham27}
D.R. Yafaev, \textquotedblleft {On the theory of the discrete spectrum of the
three-particle Schr\"odinger operator}\textquotedblright, Mat. Sb. \textbf{94}:136,
567--593, (1974).


\bibitem{ham28}
Y.N. Ovchinnikov, I.N. Sigal, \textquotedblleft {Number of bound
states of three-body systems and Efimov's effect}\textquotedblright, Ann. Phys.
\textbf{123}:2, 274--295, (1979).





\bibitem{ham29}
H. Tamura, \textquotedblleft {Asymptotic
distribution of negative eigenvalues for three-body systems in two
dimensions: Efimov effect in the antisymmetric space}\textquotedblright, Rev. Math.
Phys. \textbf{31}:9, 1950031,  (2019).

\bibitem{ham30}
S.N.Lakaev, A.K.Motovilov, S.Kh.Abdukhakimov, \textquotedblleft Two-fermion lattice Hamiltonian with first and second nearest-neighboring-site interactions\textquotedblright, J. Phys. A: Math. \textbf{56}, 315202, (2023).


\bibitem{ham31}
B.A. Lippmann, J. Schwinger, \textquotedblleft Variational principles for scattering
processes\textquotedblright, I. Phys. Rev. {\bf 79}, 361--379, (1950).

\bibitem{ham32}
S. Albeverio, F. Gesztesy, R. Khoegh-Kron, H. Holden, \textquotedblleft Solvable
Models in Quantum Mechanics\textquotedblright, Springer, New York (1988).



\bibitem{ham33}
S.N. Lakaev, I.N. Bozorov, \textquotedblleft {The number of bound
states of a one-particle Hamiltonian on a three-dimensional
lattice}\textquotedblright, Theoret. and Math. Phys. {\bf 158}, 360--376, (2009).


\bibitem{ham34}
 S.N.Lakaev, Sh.Yu.Kholmatov, Sh.I.Khamidov, \textquotedblleft
{Bose-Hubbard model with on-site and nearest-neighbor interactions;
exactly solvable case}\textquotedblright, J. Phys. A: Math. Theor. \textbf{54}, 245201, (2021).


\end{thebibliography}
\end{document}